\documentclass[journal,draftcls,onecolumn,12pt,twoside]{IEEEtran}
\usepackage{algorithm} 
\usepackage{algpseudocode}
\usepackage{graphicx}
\usepackage{amssymb}
\usepackage{amsfonts}
\usepackage{amsmath}
\usepackage{float}
\usepackage{url}
\usepackage{bm}
 \usepackage[colorlinks=true]{hyperref}
\hypersetup{urlcolor=blue, citecolor=blue}

\IEEEoverridecommandlockouts
\newtheorem{Theorem}{Theorem}
\newtheorem{Proposition}{Proposition}
\newtheorem{Corollary}{Corollary
}
\newtheorem{Definition}{Definition}
\newtheorem{Example}{Example}

\newtheorem{lemma}{Lemma}

\newenvironment{proof}{\textit{Proof\,:}} { $\square$}

\ifCLASSINFOpdf
\else
\fi
\begin{document}
\title{Efficient Search of QC-LDPC Codes with Girths 6 and 8 and Free of Elementary Trapping Sets with Small Size}
\author{ Farzane Amirzade and Mohammad-Reza~Sadeghi\\
\thanks{%
Manuscript received May ??, ????; revised November ??, ????.}
\thanks{  M.-R. Sadeghi is with the Department of Mathematics and Computer Science, Amirkabir University of Technology and F. Amirzade is with the Department of Mathematics, Shahrood University of Technology

(e-mail:  msadeghi@aut.ac.ir, famirzade@gmail.com).}
 \thanks{%
 Digital Object Identifier ????/TCOMM.?????}}


\maketitle
\begin{abstract}
One of the phenomena that influences significantly  the performance of  low-density parity-check codes is known as trapping sets. An $(a,b)$ elementary trapping set, or simply an ETS where $a$ is the size and $b$ is the number of degree-one check nodes and $\frac{b}{a}<1$,  causes high decoding failure rate and exert a strong influence on the error floor. In this paper, we provide sufficient conditions for exponent matrices to have fully connected $(3,n)$-regular QC-LDPC codes with girths 6 and 8 whose Tanner graphs are free of small ETSs. Applying sufficient conditions on the exponent matrix to remove some 8-cycles results in removing all 4-cycles, 6-cycles as well as some small elementary trapping sets. For each girth we obtain a lower bound on the lifting degree and present exponent matrices with column weight three whose corresponding Tanner graph is free of certain ETSs. 

  \end{abstract}
\begin{IEEEkeywords}
LDPC codes, girth, Tanner graph, Trapping set.
\end{IEEEkeywords}

%
\IEEEpeerreviewmaketitle
\section{Introduction}
\IEEEPARstart{Q}  uasi-cyclic low-density parity-check codes (QC-LDPC codes) are an essential  category of LDPC codes that are preferred to other types of LDPC codes because of their practical and simple implementations. One of the main approaches for constructing LDPC codes is graph-theoretic-based whose most well-known methods are progressive edge growth (PEG) and  protograph-based methods.  One of the most important representations of codes is Tanner graph. The length of the shortest cycles of the Tanner graph, girth, has been known to influence the code performance. 

 Another phenomenon that influences significantly  the performance of  binary low-density parity-check codes is known as $trapping$ $sets$. An $(a,b)$ trapping set of size $a$ is an induced subgraph of the Tanner graph on $a$ variable nodes and $b$ check nodes of odd degrees. According to the empirical results in $\cite{2014}$, among all trapping sets, the most harmful ones are those with check nodes of degree 1 or 2. This category is so-called elementary trapping sets (or simply ETSs). In addition, according to $\cite{2011}$ , $(a,b)$ ETSs that cause high decoding failure rate and exert a strong influence on the error floor are those which  satisfy the inequality $\frac{b}{a}<1$.  By increasing the girth, the lower bound on the size of trapping sets will increase, $\cite{farzane}$.  In  $\cite{2011}$ it was proved that a binary $(\gamma,\lambda)$-regular LDPC code whose Tanner graph has girth 6 contains no $(a,b)$ trapping sets of size $a\leq \gamma$, where $\frac{b}{a}<1$.  The tightest lower bound on the size of ETSs of variable-regular LDPC codes with different girths were provided in $\cite{farzane}$. For variable-regular LDPC codes with column weight $\gamma$ and girth eight the minimum size of $(a,b)$ ETSs with $\frac{b}{a}<1$ is $2\gamma-1$.  To increase the girth of Tanner graph for a given column weight, $\gamma$, one needs either to increase the number of variable nodes which results in  a code with longer length or to increase the number of check nodes and simultaneously  to decrease the row weight, $\lambda$, which result in a code with lower rate. Therefore, if the goal is to obtain an LDPC code with a certain length and rate, then the girth can not be large enough for the Tanner graph to contain no harmful trapping sets. These trapping sets characterize the size of the smallest error patterns, which can not be corrected by the decoder, as well as the slope of curve of the performance, \cite{Vasic2}.  

  Assuming $a\leq8$ and $\frac{b}{a}<1$, a characterization of $(a,b)$ trapping sets of $(3,\lambda)$-regular LDPC codes from Steiner triple systems was studied in $\cite{2010}$. A database of ETSs called the traping set ontology was introduced by Nguyen et al. in $\cite{Vasic2}$. Many efforts have been put into avoiding small trapping sets to reduce the error floor of $(3,\lambda)$-regular LDPC codes. Progressive-edge-growth method was used in $\cite{Vasic2}$ to construct  QC-LDPC codes, whose permutation matrices obtained from Latin Squares form a finite field under some matrix operation, and Tanner graph are free of some small trapping sets. In this method, $\tau$ specifies graphical structures which have to be avoided to present in the Tanner graph. The Tanner graph is built progressively in  $\lambda$ steps. The row weight, $\lambda$, is not determined and the goal is to construct a code with a rate as high as possible. In each step, a set of $N$ variable nodes, where $N$ is the lifting degree, are chosen to be added to the check nodes in a way that they correspond to $\lambda$ permutation matrices in the parity-check matrix. The terms corresponding to these permutation matrices belong to $GF(N)$. If these permutation matrices do not provide a graphical structure belonging to $\tau$, then their corresponding terms in $GF(N)$ are considered as elements in a column of  the exponent matrix, $B$. An improved  Progressive-edge-growth algorithm was presented in $\cite{Vasic}$, to construct $(3,\lambda)$-regular LDPC code whose Tanner graph has girth 8, is free of $(5,3)$ trapping sets and contains a minimum number of $(6,4)$ trapping sets. By analyzing the relationship between 8-cycles and small trapping sets in the Tanner graph  of a fully connected girth eight $(3,n)$-regular QC-LDPC codes, it was proved that controlling and avoiding some 8-cycles in the Tanner graph with girth 8 result in codes whose Tanner graphs are free of $(a,b)$ ETSs, where $a\leq8$ and $b\leq3$, $\cite{main}$.    

In this paper, we define two matrices named as  ``difference matrices", denoted by $D$ and  $DD$, from an exponent matrix of a fully connected $(3,n)$-regular QC-LDPC code.  We provide the necessary and sufficient conditions for the difference matrices  to have a Tanner graph with girth 6 and 8. The smallest size of ETSs in an LDPC code with column weight three and girth 6 is 4, \cite{2011}. For fully connected $(3,n)$-regular QC-LDPC codes with girth 6, we provide sufficient conditions for difference matrices to have a code whose Tanner graph is free of $(4,0)$ and $(4,2)$ ETSs. Our proposed method simultaneously remove all 4-cycles  as well as small ETSs. We also prove that a fully connected $(3,n)$-regular QC-LDPC code with girth 6 is free of a $(5,1)$ ETS and an analytical lower bound on the lifting degree these codes is obtained. Moreover, in this case, we present a method to obtain exponent matrices of $(3,n)$-regular QC-LDPCs with the shortest length, where $4\leq n\leq9$. The smallest size of ETSs in an LDPC code with column weight three and girth 8 is 5, \cite{farzane}. We provide sufficient conditions for difference matrices to have fully connected $(3,n)$-regular QC-LDPC codes with girth 8 whose Tanner graph is free of $(a,b)$ ETSs, where $a\leq8$ and $b\leq3$. By applying the sufficient conditions, one does not need to consider 4-cycles and 6-cycles to have a girth-8 code. Our proposed method simultaneously remove all 4-cycles, 6-cycles as well as small ETSs. In this case, we also present exponent matrices of QC-LDPCs with the shortest length.       

The rest of the paper is organized as follows. Section \ref{II} presents some basic notations, definitions and structure of difference matrices. In Sections \ref{III} and \ref{IV}, respectively, we consider sufficient condition to have a fully connected $(3,n)$-regular QC-LDPC codes with girth 6 and 8 whose Tanner graphs are free of some small ETSs. In the last section we summarize our results. 

\section{Preliminaries}\label{II}
Let $N$ be an integer number. Consider the following exponent matrix $B=[b_{ij}]$, where $b_{ij}\in \lbrace 0,1,\cdots,N-1\rbrace$, 
\begingroup\fontsize{8.5pt}{11pt}\begin{align}\label{Rela2}
B=\left[\begin{array}{cccc}
b_{00}&b_{01}&\cdots &b_{0(n-1)}\\
b_{10}&b_{11}&\cdots &b_{1(n-1)}\\
\vdots &\vdots &\ddots &\vdots \\
b_{(m-1)0}&b_{(m-1)1}&\cdots &b_{(m-1)(n-1)}\\
\end{array}\right].
\end{align}\endgroup

The $ij$-th element of the matrix, $B$, is an integer number which is substituted by an $N\times N$ matrix $I^{b_{ij}}$. This matrix is a circulant permutation matrix (CPM) in which the single 1-component of the top row is located at the $b_{ij}$-th position and other entries of the top row are zero. The $r$-th row of the circulant permutation matrix is formed by  $r$ right cyclic shifts of the first row and clearly the first row is a right cyclic shift of the last row.  The null space of the parity-check matrix provides us with a QC-LDPC code.

The necessary and sufficient condition for the existence of cycles of the length $2k$  in  the Tanner graph of QC-LDPC codes was provided in $\cite{2004}$. This well-known result is our principle tool and we summarize it as follows. If
\begingroup\fontsize{8.5pt}{11pt}\begin{align}\label{Equation}
\sum_{i=0}^{k-1}(b_{m_in_i}-b_{m_in_{i+1}})=0  \mod N,
\end{align}\endgroup
where $n_k=n_0,\ m_i\neq m_{i+1},\ n_i\neq n_{i+1}$ and $b_{m_in_i}$ is the $(m_i,n_i)$-th entry of $B$, then the Tanner graph of the parity-check matrix has cycles of the length $2k$. Equation (\ref{Equation}) proves that the cycle distribution of a code is fully described by its exponent matrix and lifting degree. If the goal is to find a QC-LDPC code with girth, $g$, and the lifting degree, $N$, from the exponent matrix, $B$, one will have to be able to find the elements $b_{ij}\in\{0,\dots,N-1\}$
	such that Equation (\ref{Equation}) are avoided for values of $k < \frac{g}{2}$.

{Equation (\ref{Equation}) is applicable to the exponent matrix. In this paper, we consider fully connected $(3,n)$-regular QC-LDPC codes and in order to simplify considering $2k$-cycles in a $3\times n$ exponent matrix we define difference matrices, named as $D$ and $DD$, then we obtain an equivalence of  Equation (\ref{Equation}) which is applicable to these two matrices for cycles of the length $2k$.

\begin{Definition}\label{Def1}
	Suppose $B$ is an $3\times n$ exponent matrix whose elements are $b_{ij},\ 0\leq i\leq 2$ and $0\leq j\leq n-1$.  The difference matrix, $D$, is defined as follows:
	\begingroup\fontsize{8.5pt}{11pt}\begin{align}\label{Rela1}
	D=\left[\begin{array}{lllll}
	b_{00}-b_{10} & \dots & b_{0(n-1)}-b_{1(n-1)}\\
	b_{00}-b_{20}  & \dots & b_{0(n-1)}-b_{2(n-1)}\\
	b_{10}-b_{20}  & \dots & b_{1(n-1)-b_{2(n-1)}}\\
	\end{array}\right].
	\end{align}\endgroup
\end{Definition}
\noindent We also utilize another matrix to reduce the complexity. It is obtained from $D$ which we denote it by $DD$.
\begin{Definition}\label{Def2}
		A $3\times{n\choose 2}$ difference matrix $DD$ is constructed by subtracting every two columns of $D$ as follows. Suppose  $D_{ij}$ and $D_{ij'}$ are two elements of the difference matrix, $D$, which occur in the same row, $i$, and disjoint columns,  $j$ and $j'$, respectively, where $j<j'$. If we subtract $j'$-th column of $D$ from $j$-th column of $D$, then  $(D_{ij}-D_{ij'},D_{ij'}-D_{ij})\mod N$ is defined as an element of the $i$-th row and $(jj')$-th column of $DD$.
\end{Definition}
\begin{Example}\label{Example1}
	Let $B$ be a $3\times4$ exponent matrix with the lifting degree $N=37$ as follows:
	\begingroup\fontsize{8.5pt}{11pt}\begin{align}
	B=\left[\begin{array}{cccc}
	0 & 0 & 0 & 0 \\
	0 & 1 & 3 & 24 \\
	0 & 27 & 7 & 19 \\
	\end{array}\right]. 
	\end{align}\endgroup
	According to Definitions \ref{Def1} and \ref{Def2} we construct two difference matrices $D$ and $DD$:
	\begingroup\fontsize{8.5pt}{11pt}\begin{align}
	D=\left[\begin{array}{cccc}
	0 & -1 & -3 & -24\\
	0 & -27 & -7 & -19\\
	0 & -26& -4 & 5 \\
	\end{array}\right],
	\end{align}\endgroup
	\begingroup\fontsize{8.5pt}{11pt}\begin{align}
	DD=\left[\begin{array}{cccccc}
	(1,36) & (3,34) & (24,13) & (2,35) & (23,14) & (21,16)\\
	(27,10) & (7,30) & (19,18) & (17,20) & (29,8) & (12,25) \\
	(26,11) & (4,33) & (32,5) & (15,22) & (6,31) & (28,9) \
	\end{array}\right].
	\end{align}\endgroup
\end{Example}
Take an induced subgraph of the Tanner graph on a subset $S$ of $V$. The subgraph contains some check nodes of odd degrees and some check nodes of even degrees referred to as  unsatisfied check nodes and satisfied check nodes, respectively. If $|S|=a$ and the number of unsatisfied check nodes is $b$, then the induced subgraph provides an $(a,b)$ trapping set of size $a$. An $(a,b)$ trapping set is called elementary if all check nodes are of degree 1 or 2. As a result, all unsatisfied check nodes in an elementary trapping set (or ETS) are of degree one. 

In $\cite{Amiri}$, for a bipartite graph $G$ corresponding to
an elementary absorbing set, a $variable\ node\ (VN)\ graph$ is
constructed by removing all degree-one check nodes, defining variable nodes of G as its vertices and degree-two check nodes connecting the variable nodes in $G$ as
its edges. We use this graph representation for an $(a,b)$ ETS.

The existence of a $2k$-cycle in the ETS is equivalent to the existence of a cycle of length $k$ in its correspondent $VN$ graph. For example, if a sequence of $v_0,c_0,v_1,c_1,v_2,c_2$ is a 6-cycle of an ETS, where $v_i\in V$ and $c_i\in C$, then by replacing any check node with an edge we obtain a cycle of length three in the $VN$ graph whose vertices are $v_1,v_2,v_3$.  Moreover, any 4-cycle in the trapping set is equivalent to a multiple edge in its corresponding $VN$ graph. As an example, if a sequence of $v_0,c_0,v_1,c_1$ is a 4-cycle of a trapping set, then by replacing any degree-two check node with an edge we have a multiple edge $(v_1,v_2)$.

In a 4-cycle free Tanner graph, the $VN$ graph of each ETS is free of multiple edges which is called a simple graph. And in a Tanner graph with girth at least 8, the $VN$ graph of each ETS is a simple and triangle-free graph. For example, in Fig. \ref{FIG1} the ETS contains no 6-cycle and its corresponding $VN$ graph is triangle-free. Variable nodes, satisfied and unsatisfied check nodes are denoted by circles, empty squares  and full squares, respectively.
\begin{center}
\begin{figure}
\centering
\includegraphics[scale=.3]{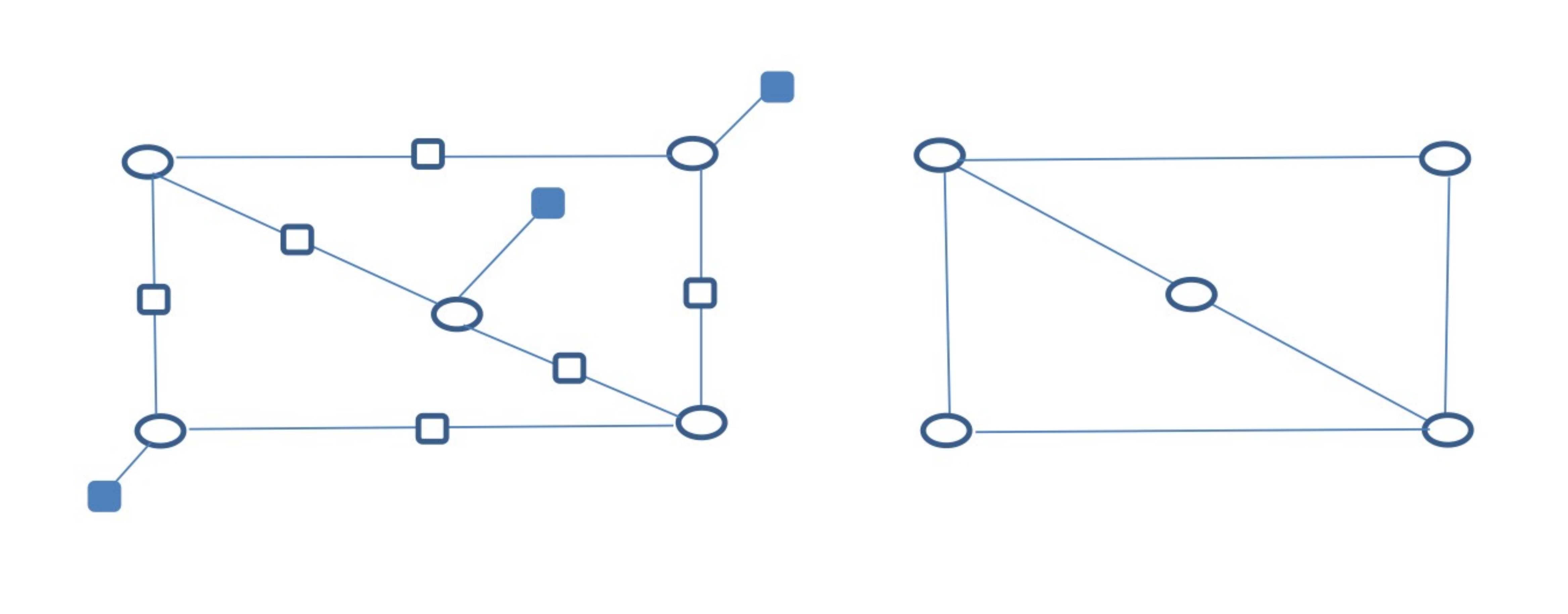}
\caption{ A (5,3) elementary trapping set with $\gamma=4$ and its  corresponding  variable node graph}\label{FIG1}
\end{figure}
\end{center}
\section{Construction of $(3,n)$-regular QC-LDPC codes with girth 6 and free of ETSs with small size}\label{III}

In this section, we first consider necessary and sufficient condition for difference matrices to have Tanner graph with girth 6. Then, we provide sufficient conditions for exponent matrices to obtain  fully connected $(3,n)$-regular  QC-LDPC codes with girth 6 and free of $(4,0)$ and $(4,2)$ ETSs and we  prove fully connected $(3,n)$-regular  QC-LDPC codes with girth 6 are free of $(5,1)$ ETSs. 

In order to consider 4-cycles, Equation (\ref{Equation}) has to be investigated for every $2\times2$ submatrix of the exponent matrix. Consider a $2\times2$ submatrix  of the exponent matrix in two rows $i_1$ and $i_2$ and two columns $j_1$ and $j_2$. If the submatrix leads to 4-cycles in the Tanner graph, then Equation (\ref{Equation}) gives $(b_{i_1j_1}-b_{i_1j_2})+(b_{i_2j_2}-b_{i_2j_1})=0\ \mod N$. But, in order to consider 4-cycles using the difference matrix, $D$, we rearrange the left side of the equality as follows:
\begin{center}
	$(b_{i_1j_1}-b_{i_2j_1})-(b_{i_1j_2}-b_{i_2j_2})=0 \mod N.$
\end{center}

 The expression $(b_{i_1j_1}-b_{i_2j_1})$ is an element of $D$ in the $i$-th row and the $j_1$-th column and $(b_{i_1j_2}-b_{i_2j_2})$ is another element of $D$ in the  $i$-th row and the $j_2$-th column. So, we conclude that if $b_{i_1j_1}-b_{i_2j_1}=b_{i_1j_2}-b_{i_2j_2} \mod N$, or equivalently if $D_{ij1}=D_{ij2}$, then the Tanner graph has 4-cycles. Moreover, $\pm (D_{ij1}-D_{ij2})\mod N$ is an element of the difference matrix, $DD$. As a result, every integer number in the  difference matrix, $DD$, is equivalent to the result of Equation (\ref{Equation}) to consider 4-cycles and	Tanner graph is 4-cycle free if and only if the difference matrix, $DD$, has no zero element.

 Consider a fully connected $(3,n)$-regular  QC-LDPC code with an exponent matrix, $B$. Every vertex of the $VN$ graph corresponds to a column of $B$ and each edge of the $VN$ graph corresponds to a row of $B$. Degree of each vertex determines the number of rows of the exponent matrix which are involved in an ETS. Suppose each edge of the $VN$ graph is characterized with a row index of $B$, which we denote them by $u,v,w$. Assuming such $VN$ graphs with a label for each edge we obtain sufficient conditions for difference matrix, $DD$, to have a fully connected $(3,n)$-regular  QC-LDPC code free of a given ETS.

The number of non-isomorphic $(4,0)$ ETSs in a Tanner graph with girth 6 is one whose $VN$ graph is a complete graph with 4 vertices. As we see in Fig. \ref{FIG2}, it contains 4-cycles whose edge labels belong to two rows of the exponent matrix.  Each 4-cycle of the $VN$ graph corresponds to an 8-cycle of the Tanner graph. Therefore, in order to construct  a fully connected $(3,n)$-regular  QC-LDPC code with girth 6 which is free of a $(4,0)$ ETS, it is sufficient to consider Equation (\ref{Equation}) for $2\times2,\ 2\times 3$ and $2\times4$ submatrices of the exponent matrix to avoid 8-cycles.
\begin{center}
	\begin{figure}
		\centering
		\includegraphics[scale=.3]{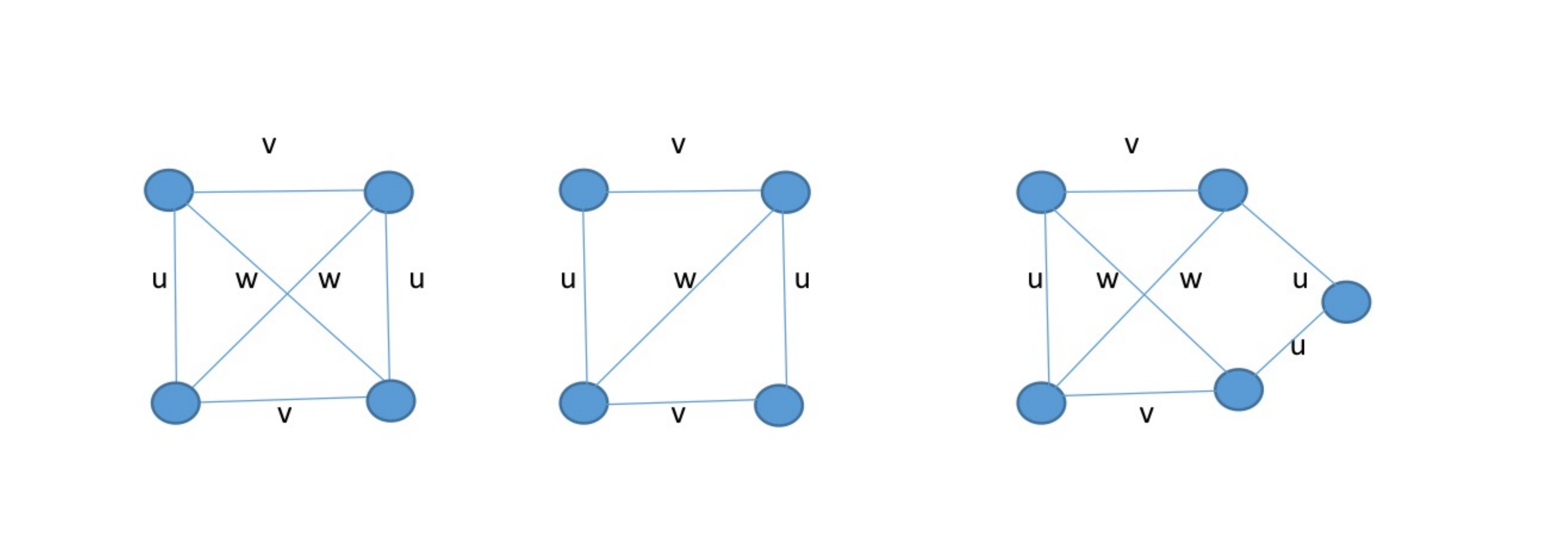}
		\caption{ The  variable node graphs of $(4,0),\ (4,2)$ and $(5,1)$ ETSs with girth 6}\label{FIG2}
	\end{figure}
	\end{center}
\begin{Theorem}\label{Theorem1}
	The sufficient conditions for  exponent matrix to have a fully connected $(3,n)$-regular  QC-LDPC code with girth 6 which is free of $(4,0)$ and $(4,2)$ ETSs are the non-existence of repeated elements and zero elements in each row of the difference matrix, $DD$, and $2\times DD\neq0\mod N$. 
\end{Theorem}
\begin{proof}
We consider necessary and sufficient conditions  for $DD$ to avoid 8-cycles which contain two rows of the exponent matrix.
\begin{itemize}
	\item  Take  $j_0$ and $j_1$ as two column indices of a  $2\times2$ submatrix of $B$. The left side of Equation (\ref{Equation}) is $b_{i_0j_0}-b_{i_0j_1}+b_{i_1j_1}-b_{i_1j_0}+b_{i_0j_0}-b_{i_0j_1}+b_{i_1j_1}-b_{i_1j_0}=2(b_{i_0j_0}-b_{i_0j_1}+b_{i_1j_1}-b_{i_1j_0})$.
	By rearranging the terms of the equation we have  $2((b_{i_0j_0}-b_{i_1j_0})-(b_{i_0j_1}-{i_1j_1}))=2(D_{ij_0}-D_{ij_1}),$ where $i\in\{0,1,2\}$. So, to avoid 8-cycle in this case we have to consider the inequality $2(D_{ij_0}-D_{ij_1})\neq0$. Since $(D_{ij_0}-D_{ij_1})$ is a component of an element of the matrix, $DD$, we have $2DD\neq 0$. 
	\item   Take $j_0,j_1,j_2$ as three column indices of a $2\times3$ submatrix of $B$. The left side of Equation (\ref{Equation}) and its corresponding expression whose elements belong to $D$ are as follows:
2	
	$b_{i_0j_0}-b_{i_0j_1}+b_{i_1j_1}-b_{i_1j_2}+b_{i_0j_2}-b_{i_0j_1}+b_{i_1j_1}-b_{i_1j_0}=(b_{i_0j_0}-b_{i_1j_0})-(b_{i_0j_1}-b_{i_1j_1})+(b_{i_0j_2}-b_{i_1j_2})-(b_{i_0j_1}-b_{i_1j_1})=D_{ij_0}-D_{ij_1}+D_{ij_2}-D_{ij_1}$.
	
	So, to avoid 8-cycles in this case,   the difference matrix of QC-LDPC code have to satisfy in the inequality $D_{ij_0}-D_{ij_1}+D_{ij_2}-D_{ij_1}\neq0$ or equivalently, $D_{ij_0}-D_{ij_1}\neq D_{ij_1}-D_{ij_2}$. To obtain this inequality the 8-cycle is started from $b_{i_0j_0}$. If the cycle is started from $b_{i_0j_1}$, then we obtain one of the inequalities $D_{ij_0}-D_{ij_1}\neq -(D_{ij_0}-D_{ij_2})$ or $D_{ij_0}-D_{ij_2}\neq -(D_{ij_1}-D_{ij_2})$. Note that two sides of inequalities are components of two elements of the $i$-th row of the matrix $DD$.
	\item  Take  $j_0,j_1,j_2$ and $j_3$ as four column indices of a $2\times4$ submatrix of $B$. Like the previous item by investigating Equation (\ref{Equation}) and their equivalences in the difference matrix we conclude that the difference matrix of a QC-LDPC code which is free of 8-cycles holds the inequality: $\pm(D_{ij_0}-D_{ij_1})\neq \pm(D_{ij_2}-D_{ij_3})$. Note that two sides of inequalities are components of two elements of the $i$-th row of the matrix $DD$. 
\end{itemize}
The above items prove that to avoid 8-cycles which are related to two rows of the exponent matrix it is necessary to have a difference matrix, $DD$, which contain $2{n\choose 2}$ disjoint non-zero elements in each row and $2DD$ contains no zero elements. Since both $(4,0)$ and $(4,2)$ ETSs have 8-cycles which are related to two rows of the exponent matrix, by removing such 8-cycles these two types of ETSs will be removed too.
\end{proof}
\begin{Corollary}
The lower bound on the lifting degree of fully connected $(3,n)$-regular  QC-LDPC code with girth 6 which is free of $(4,0)$ and $(4,2)$ ETSs is $2{n\choose 2}=n^2-n$. 
\end{Corollary}	
\begin{Proposition} 
	A fully connected $(3,n)$-regular  QC-LDPC code with girth 6 has no $(5,1)$ ETS.
\end{Proposition}
\begin{proof}
	The number of non-isomorphic $(5,1)$ ETSs in a Tanner graph with girth 6 is one. Its corresponding $VN$ graph is shown in Fig. \ref{FIG2} As we see, there is a vertex in the the $VN$ graph of a $(5,1)$ ETS which contains two edge with the same row index. It indicates that a column of the parity-check matrix contains two 1-components which belong to a CPM which contradicts with the definition of CPM. 
\end{proof}	

In Table \ref{Tabel} we provide exponent matrices of $(3,n)$-regular  QC-LDPC codes with girth 6 and the shortest length which are free of $(4,0)$, $(4,2)$ and $(5,1)$ ETSs. In order to reduce the size of the search space we assume the first row and the first column are all-zero which are omitted in the Table. In addition, the third row is the multiplication of the second row by 2.

\begin{table}[h]
	\begin{center}
		\caption{Fully connected $(3,n)$-regular QC-LDPC codes with girth 6 and $4\leq n\leq 9$ whose Tanner graphs are free of $(4,0)$, $(4,2)$ and $(5,1)$ ETSs}\label{Tabel}
		\begin{tabular}{|c|c|c|c|c|}
			\hline
			$n,\ N$&$ exponent\ matrices$&$n,\ N$&$ exponent\ matrices$\\
			\hline
			$n=4, N=13$ & $\begin{array}{ccc}
			1 & 3 & 9 \\
			2& 6& 5  \\
			\end{array}$&$n=7, N=49$&  $\begin{array}{cccccc}
			1 & 3 & 7 & 27&35&40\\
			2&6&14&5&21&31
			\end{array}$\\
			\hline
			$n=5, N=21$& $\begin{array}{cccc}
			1 & 4 & 14 & 16\\
			2&8&7&11
			\end{array}$&$n=8, N=57$&  $\begin{array}{ccccccc}
			1 & 3 & 13 & 32&36&43&52\\
			2&6&26&7&15&29&47
			\end{array}$\\
			\hline
			$n=6,N=31$&  $\begin{array}{ccccc}
			1 & 3 & 8 & 12&18\\
			2&6&16&24&5
			\end{array}$ &$n=9, N=85$&  $\begin{array}{cccccccc}
			1 & 4 & 12 & 14&19&35&41&61\\
			2&8&24&28&38&70&82&37
			\end{array}$\\
			\hline
		\end{tabular}
	\end{center}
\end{table}

\section{Construction of $(3,n)$-regular QC-LDPC codes with girth 8 and free of ETSs with small size}\label{IV}

In this section, we aim to construct a $(3,n)$-regular QC-LDPC code with girth 8 whose Tanner graph is free of $(a,b)$ ETSs, where $a\leq8$ and $b\leq3$.  As explained in $\cite{main}$, assuming the girth is 8, if $u,v,w$ are three row indices of the exponent matrix, then removing 8-cycles on $3\times2$, $3\times3$ and $3\times4$ submatrices of the exponent matrix which have two check nodes with the same row index causes to avoid $(5,3)$ ETSs. The $VN$ graph of such ETS has 4-cycles  whose edge labels are one the following sets, $\{u,v,u,w\},\{v,u,v,w\},\{w,v,w,u\}$. Avoiding one of these 8-cycle cause to remove $(5,3)$ ETSs. Avoiding  $(5,3)$ and $(6,4)$ ETSs cause to remove $(7,3)$ ETSs, $\cite{Vasic2}$. In order to remove $(7,3)$ ETSs, both 8-cycles  which have two check nodes with the same row index and 8-cycles on two rows of the exponent matrix have to be avoided. The edge labels of a 4-cycle in  the $VN$ graph which corresponds to an 8-cycle with two row indices of the exponent matrix belong to one of the sets $\{u,v,u,v\},\{u,w,u,w\},\{v,w,v,w\}$. In fact, if one chooses to avoid 8-cycles with edge labels $\{v,u,v,w\}$ in its corresponding $VN$ graph and 8-cycles with  edge labels $\{u,v,u,v\}$ and $\{v,w,v,w\}$, then  $(5,3)$, $(6,4)$ and $(7,3)$ ETSs will be removed. Generally, if we choose to remove all 8-cycles whose edge labels in the $VN$ graph belong to one of the following three sets \\
$(i)\ \{u,v,u,v\},\{u,w,u,w\},\{u,v,u,w\}$, \\
$(ii)\ \{u,v,u,v\},\{v,w,v,w\},\{v,w,v,u\}$ or \\
$(iii)\ \{w,v,w,v\},\{u,w,u,w\},\{v,w,u,w\}$, \\
then Tanner graph is free of $(a,b)$ ETSs, where $a\leq8$ and $b\leq3$,  $\cite{main}$.  Authors in $\cite{main}$ chose to remove 8-cycles whose edge labels in the VN graph belong to the three sets $(ii)$.  In this structure all of 4-cycles and 6-cycles in addition to the mentioned 8-cycles have to be avoided to obtain the desired girth-8 QC-LDPC code. 

In this section, we show that removing all 8-cycles whose edge labels in the $VN$ graph belong to three sets $(i)$ results in  $(3,n)$-regular QC-LDPC codes with girth 8 whose Tanner graph is free of $(a,b)$ ETSs, where $a\leq8$ and $b\leq3$. We use the difference matrix, $DD$, and we prove that none of  6-cycles are required to be considered in this method. In fact, by avoiding  the mentioned 8-cycles all of 6-cycles will be removed too. In order to prove our claim we have to provide equivalence of Equation (\ref{Equation}) for 6-cycles whose terms belong to the difference matrices.
  
\begin{lemma}\label{lemma4}
	Let $DD$ be a difference matrix corresponding to an $3\times n$ exponent matrix, $B$ whose first row and column are all-zero. Take $DD_{ij}$ and $N-DD_{ij}$ as the first and the second components of the $ij$-th  element of $DD$, respectively.  If the Tanner graph is 6-cycle free, then a $3\times 3$ submatrix of $DD$ fulfills the following inequalities:
	\begingroup\fontsize{8.5pt}{11pt}\begin{align}\label{7}
	\begin{array}{ll}
	1)\ DD_{0(j_0,j_1)}\neq DD_{1(j_0j_2)}&	2)\ DD_{0(j_0,j_1)}\neq N-DD_{1(j_1j_2)} \\ 3)\ DD_{0(j_0j_2)}\neq DD_{1(j_1j_2)} & 4)\ DD_{0(j_0j_2)}\neq DD_{1(j_0j_1)} \\  5)\ DD_{0(j_1j_2)}\neq N-DD_{1(j_0j_1)} &6)\ DD_{0(j_1j_2)}\neq DD_{1(j_0j_2)},
	\end{array}
	\end{align}\endgroup	 
	\noindent Disjoint column indices $j_{0},j_{1}$ and $j_{2}$ of $DD$ are corresponding to the three columns of $B$.
\end{lemma}   
\begin{proof}
	We prove by applying Equation (\ref{Equation}) in two types of $3\times 3$ submatrix of $B$. In the first type, which we denote it by $B'$, the first row and column are all-zero. In the second type, which we denote it by $B''$, the first column is not all-zero. Suppose \begingroup\fontsize{8.5pt}{11pt}\begin{align}
	B'=\left[\begin{array}{ccc}
	0 & 0 & 0  \\
	0 & b_{1j_1} & b_{1j_2}  \\
	0 & b_{2j_1} & b_{2j_2} \\
	\end{array}\right],\ \ \ B''=\left[\begin{array}{ccc}
	0 & 0 & 0  \\
	b_{1j_0} & b_{1j_1} & b_{1j_2}  \\
	b_{2j_0} & b_{2j_1} & b_{2j_2} \\
	\end{array}\right]. 
	  	\end{align}\endgroup
By applying Definitions \ref{Def1} and \ref{Def2} in $B'$ and $B''$ we have the following submatrices of the  difference matrix, $DD$, which we denote them by $DD'$ and $DD''$, respectively. Since we only need the first and the second rows of these two submatrices to prove Lemma, we present these two rows for each submatrix.\begingroup\fontsize{8.5pt}{11pt}\begin{align}
DD'=\left[\begin{array}{ccc}
 (b_{1j_1},N-b_{1j_1}) & (b_{1j_2},N-b_{1j_2})& (b_{1j_2}-b_{1j_1},N-b_{1j_2}+b_{1j_1})  \\
 (b_{2j_1},N-b_{2j_1}) & (b_{2j_2},N-b_{2j_2})& (b_{2j_2}-b_{2j_1},N-b_{2j_2}+b_{2j_1})\\
\end{array}\right], 
\end{align}\endgroup
\begingroup\fontsize{8.5pt}{11pt}\begin{align}
DD''=\left[\begin{array}{ccc}
(b_{1j_1}-b_{1j_0},N-b_{1j_1}+b_{1j_0}) & (b_{1j_2}-b_{1j_0},N-b_{1j_2}+b_{1j_0})& (b_{1j_2}-b_{1j_1},N-b_{1j_2}+b_{1j_1})  \\
(b_{2j_1}-b_{2j_0},N-b_{2j_1}+b_{2j_0}) &(b_{2j_2}-b_{2j_0},N-b_{2j_2}+b_{2j_0})& (b_{2j_2}-b_{2j_1},N-b_{2j_2}+b_{2j_1})\\
\end{array}\right], 
\end{align}\endgroup	  
Using Equation (\ref{Equation}) for $B'$ gives the following inequalities to avoid 6-cycles. We present their equivalences whose elements belong two the submatrix, $DD'$, of the difference matrix, $DD$,
\begingroup\fontsize{8.5pt}{11pt}\begin{align}\label{11}
\begin{array}{l}
{B'}_{00}-{B'}_{01}+{B'}_{11} - {B'}_{12}+{B'}_{22}-{B'}_{20}=N-{DD'}_{02}+{DD'}_{11}\neq0\\
{B'}_{00}-{B'}_{02}+{B'}_{12} - {B'}_{11}+{B'}_{21}-{B'}_{20}={DD'}_{02}+{DD'}_{10}\neq0 \\
{B'}_{01}-{B'}_{02}+{B'}_{12} - {B'}_{10}+{B'}_{20}-{B'}_{21}={DD'}_{01}+N-{DD'}_{10}\neq0\\
{B'}_{01}-{B'}_{00}+{B'}_{10} - {B'}_{12}+{B'}_{22}-{B'}_{21}={DD'}_{12}+N-{DD'}_{01}\neq0\\
{B'}_{02}-{B'}_{00}+{B'}_{10} - {B'}_{11}+{B'}_{21}-{B'}_{22}=N-{DD'}_{12}+N-{DD'}_{00}\neq0 \\
{B'}_{02}-{B'}_{01}+{B'}_{11} - {B'}_{10}+{B'}_{20}-{B'}_{22}={DD'}_{00}+N-{DD'}_{11}\neq0.\\
\end{array}
\end{align}\endgroup
Using Equation (\ref{Equation}) for $B''$ gives the following inequalities to avoid 6-cycles. We present their equivalences whose elements belong two the submatrix, $DD''$, of the difference matrix, $DD$,
\begingroup\fontsize{8.5pt}{11pt}\begin{align}\label{12}
\begin{array}{l}
{B''}_{00}-{B''}_{01}+{B''}_{11} - {B''}_{12}+{B''}_{22}-{B''}_{20}=N-{DD''}_{02}+{DD''}_{11}\neq0\\
{B''}_{00}-{B''}_{02}+{B''}_{12} - {B''}_{11}+{B''}_{21}-{B''}_{20}={DD''}_{02}+{DD''}_{10}\neq0 \\
{B''}_{01}-{B''}_{02}+{B''}_{12} - {B''}_{10}+{B''}_{20}-{B''}_{21}={DD''}_{01}+N-{DD''}_{10}\neq0\\
{B''}_{01}-{B''}_{00}+{B''}_{10} - {B''}_{12}+{B''}_{22}-{B''}_{21}={DD''}_{12}+N-{DD''}_{01}\neq0\\
{B''}_{02}-{B''}_{00}+{B''}_{10} - {B''}_{11}+{B''}_{21}-{B''}_{22}=N-{DD''}_{12}+N-{DD''}_{00}\neq0 \\
{B''}_{02}-{B''}_{01}+{B''}_{11} - {B''}_{10}+{B''}_{20}-{B''}_{22}={DD''}_{00}+N-{DD''}_{11}\neq0.\\
\end{array} 
\end{align}\endgroup 
If we consider the above inequalities \ref{11} and  \ref{12} for each $3\times3$ submatrix of $B$, then we obtain the inequalities \ref{7}.   	  	
\end{proof}
\begin{lemma}\label{lemma2}
The  sufficient condition  for $DD$ to avoid 8-cycles which contain three rows of the exponent matrix with row indices, $u,v,w$, which have two check nodes with the same row index, $u$, is the existence of non-zero disjoint elements in the first two rows of $DD$.
\end{lemma}
\begin{proof}
	In order to consider 8-cycles on three rows of the exponent matrix, Equation (\ref{Equation}) have to be investigated for $3\times2,\ 3\times3$ and $3\times4$ submatrices of $B$. Suppose edge labels such 8-cycles in the $VN$ graph are $\{u,v,u,w\}$. In the following three items we consider inequalities to  avoid such 8-cycles.  
	\begin{enumerate}
		\item  Take $j_0,j_1$ as two column indices of a $3\times2$ submatrix of $B$. The left side of Equation (\ref{Equation}) and its corresponding expressions whose elements belong to $D$ are as follows:
		
		$b_{vj_0}-b_{vj_1}+b_{uj_1}-b_{uj_0}+b_{wj_0}-b_{wj_1}+b_{uj_1}-b_{uj_0}=-(b_{uj_0}-b_{vj_0})+(b_{uj_1}-b_{vj_1})-(b_{uj_0}-b_{wj_0})+(b_{uj_1}-b_{wj_1})=-D_{0j_0}+D_{0j_1}-D_{1j_0}+D_{1j_1}$,

		So, one of the sufficient conditions for the difference matrices to avoid  8-cycles whose edge labels in the $VN$ graph are $\{u,v,u,w\}$ is $\pm(D_{0j_0}-D_{0j_1})\neq \pm(D_{1j_0}-D_{1j_1})$ or equivalently $\pm(DD_{0(j_0j_1)})\neq \pm(DD_{1(j_0j_1)})$. 
		\item Take $3\times3$ submatrices of the exponent matrix, where $j_0,\ j_1$ and $j_2$  are three column indices. Like the previous item by investigating  Equation (\ref{Equation}) and their equivalences in the difference matrix we have: $b_{uj_0}-b_{uj_1}+b_{wj_1}-b_{wj_2}+b_{uj_2}-b_{uj_1}+b_{vj_1}-b_{vj_0}=(b_{uj_0}-b_{vj_0})-(b_{uj_1}-b_{vj_1})-(b_{uj_1}-b_{wj_1})+(b_{uj_2}-b_{wj_2})=D_{0j_0}-D_{0j_1}-D_{1j_0}+D_{1j_2}=DD_{0(j_0j_1)}+N-DD_{1(j_0j_2)}$.
		\item  Take $3\times4$ submatrices of the exponent matrix, where $j_0,\ j_1,\ j_2$ and $j_3$  are four column indices. The left side of Equation  (\ref{Equation}) is $b_{vj_0}-b_{vj_1}+b_{uj_1}-b_{uj_2}+b_{wj_2}-b_{wj_3}+b_{uj_3}-b_{uj_0}.$ In the following we rearrange the terms of the expression to obtain an expression whose terms belong $D$.
		
		$-(b_{uj_0}-b_{vj_0})+(b_{uj_1}-b_{vj_1})-(b_{uj_2}-b_{wj_2})+(b_{uj_3}-b_{wj_3})=-(D_{0j_0}-D_{0j_1})-(D_{1j_2}-D_{1j_3})=N-DD_{0(j_0j_1)}+N-DD_{1(j_2j_3)}$.

	\end{enumerate}
The three above items demonstrates that sufficient condition to avoid 8-cycles whose edge labels in the $VN$ graph are $\{u,v,u,w\}$ is the existence of non-zero disjoint elements in the first two rows of the difference matrix, $DD$.  
\end{proof}
\begin{Theorem}\label{Theorem2}
	Suppose the first row and the first column of the exponent matrix is all-zero. The sufficient conditions for the exponent matrix of a $(3,n)$-regular QC-LDPC code with girth 8 whose Tanner graph is free of $(a,b)$ ETSs, where $a\leq8$ and $b\leq3$ are as follows:
	\begin{itemize}
		\item The first two rows of the matrix, $2\times DD\mod N$, are free of zero elements,
		\item the first two rows of the difference matrix, $DD$, are free of repeated elements,
		\item the difference matrix, $DD$, is free of zero elements.
	\end{itemize}
\end{Theorem}
\begin{proof}
	Suppose the difference matrix, $DD$, fulfills the three sufficient condition. The third condition proves the non-existence of 4-cycles. Since a difference matrix, $DD$, with non-zero disjoint elements in the first two rows satisfies the inequalities \ref{7}, Tanner graph is 6-cycle free. The first and second conditions demonstrate the non-existence of 8-cycles whose edge labels in the $VN$ graph belong to three sets $(i)$. According to the proof of Theorem \ref{Theorem1}, necessary and sufficient conditions  for $DD$ to avoid 8-cycles whose edge labels belong to the set $\{u,v,u,v\}$, the first row of $DD$ has to contain non-zero disjoint elements and the matrix $2\times DD\mod N$ has no zero element in the first row. To avoid 8-cycles whose edge labels belong to the set $\{u,w,u,w\}$, the second row of $DD$ has to contain non-zero disjoint elements and the matrix $2\times DD\mod N$ has no zero element in the second row. As we proved in Lemma \ref{lemma2}, the sufficient condition to remove 8-cycles whose edge labels in the $VN$ graph belong to the set $\{u,v,u,w\}$ of $(i)$ is the non-existence of repeated and zero elements in the first two rows of the difference matrix, $DD$.  So, if the first two rows of the difference matrix, $DD$, are free of  repeated and zero elements and the first two rows of the matrix, $2DD$,  and the third row of $DD$ are free of zero elements, then Tanner graph has girth 8 and is free of small $(a,b)$ ETSs, where $a\leq8$ and $b\leq3$. 
\end{proof}

According to the proof of Theorem \ref{Theorem2}, if one of three sets $(i),(ii)$ and $(iii)$ is chosen to remove 8-cycles which results in QC-LDPC codes free of $(a,b)$ ETSs, where $a\leq8$ and $b\leq3$, then the non-existence of repeated elements in two rows of the difference matrix, $DD$, is necessary. An immediate result of this necessary condition and the fact that every row of $DD$ contains  $2{n\choose 2}$ elements is as follows.  
\begin{Corollary}\label{corollary2}
	The lower bound on the lifting degree of fully connected $(3,n)$-regular  QC-LDPC codes with girth 8 which is free of $(a,b)$ ETSs, where $a\leq8$ and $b\leq3$ is $N=4{n\choose 2}=2n(n-1)$.
\end{Corollary}

To have the non-existence of repeated elements in the first two rows of $DD$, it is sufficient to consider the smaller element of each pair of $DD$ in these two rows. If they provide a subset of $\{1,2,\dots,[\frac{N}{2}]\}$ with the cardinality $2{n\choose 2}$, then the exponent matrix satisfies in the second condition of Theorem $\ref{Theorem2}$. In Table \ref{Tabell}, we present $(3,n)$-regular QC-LDPC codes with girth 8 and the shortest length whose Tanner graph is free of small $(a,b)$ ETSs, where $a\leq8$ and $b\leq3$. To obtain these exponent matrices we use Theorem \ref{Theorem2} and Corollary \ref{corollary2}. 
\begin{table}[h]
	\begin{center}
		\caption{Fully connected $(3,n)$-regular QC-LDPC codes with girth 8 and $4\leq n\leq 9$ whose Tanner graphs are free of $(a,b)$ ETSs, where $a\leq8$ and $b\leq3$}\label{Tabell}
		\begin{tabular}{|c|c|c|c|c|} 
			\hline
			$n,\ N$&$ exponent\ matrices$\\
			\hline
			$n=4, N=26$ & $\begin{array}{ccc}
			1 & 3 & 9 \\
			4& 11& 16  \\
			\end{array}$\\
			\hline
			$n=5, N=21$& $\begin{array}{cccc}
			1& 4& 11& 29\\
			 2& 8 &17& 22
			\end{array}$\\
			\hline
			$n=6,N=31$&  $\begin{array}{ccccc}
			 1& 13& 16& 33& 39\\
			2& 7& 11& 21& 29
			\end{array}$ \\
			\hline
			$n=7, N=91$&  $\begin{array}{cccccc}
			1& 4& 13& 30& 40& 45\\
			2& 8& 22& 33& 56& 75
			\end{array}$\\
			\hline
		\end{tabular}
	\end{center}
\end{table}
\section{Conclusion}\label{V}
In this paper, we provided sufficient conditions for exponent matrices to have fully connected $(3,n)$-regular QC-LDPC codes with girths 6 and 8 whose Tanner graphs are free of small elementary trapping sets. We demonstrated that applying sufficient conditions on the exponent matrix to remove some 8-cycles results in removing all 4-cycles, 6-cycles as well as some small elementary trapping sets. For each girth we obtained a lower bound on the lifting degree and presented exponent matrices with column weight three whose corresponding Tanner graph is free of certain trapping sets.

\end{document}